\documentclass[12pt]{amsart}
\usepackage{fullpage}
\usepackage{amsmath,amssymb}
\usepackage{graphicx,subfigure,tabularx,multirow}
\usepackage{url}
\usepackage{color}
\usepackage{enumerate}

\newtheorem{theorem}{Theorem}
\newtheorem{lemma}{Lemma}

\renewcommand{\paragraph}[1]{\vspace{1mm}\noindent {\bf #1}}
\newcommand{\comm}[1]{}

\newcommand{\df}{\textbf}
\DeclareMathOperator{\inv}{inv}

\newcommand{\ZZ}{{\mathbb Z}}

\begin{document}

\title{Drawing Permutations with Few Corners}


\author{Sergey Bereg}
\author{Alexander E. Holroyd}
\author{Lev Nachmanson}
\author{Sergey Pupyrev}

\begin{abstract}
A permutation may be represented by a collection of paths in the
plane.  We consider a natural class of such representations, which we
call tangles, in which the paths consist of straight segments at 45
degree angles, and the permutation is decomposed into
nearest-neighbour transpositions.  We address the problem of
minimizing the number of crossings together with the number of corners
of the paths, focusing on classes of permutations in which both can
be minimized simultaneously. We give algorithms for computing such
tangles for several classes of permutations.
\end{abstract}

\maketitle

\section{Introduction}
\begin{figure}[b]
\centering
\includegraphics[scale=.7]{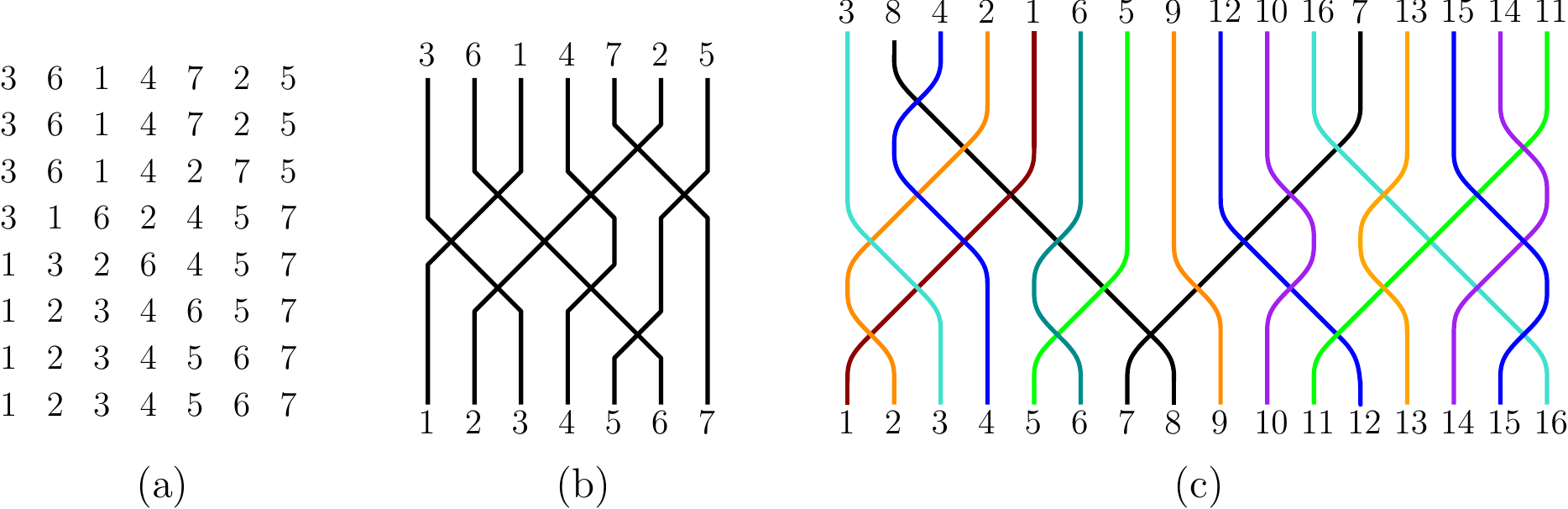}
\caption{
        (a)~A tangle solving the permutation $[3, 6, 1, 4, 7, 2, 5]$.
        (b)~A drawing of the tangle.
        (c)~An example of a perfect tangle drawing.
		(Here colors and rounded corners are employed to further enhance aesthetics and readability).}
\label{example}
\end{figure}
What is a good way to visualize a permutation? In this paper we study
drawings in which a permutation of interest is connected to the identity
permutation via a sequence of intermediate permutations, with consecutive
elements of the sequence differing by one or more non-overlapping
nearest-neighbour swaps.  The position of each permutation element through
the sequence may then traced by a piecewise-linear path comprising segments
that are vertical and $45^{\circ}$ to the vertical. Our goal is to keep these
paths as simple as possible and to avoid unnecessary crossings.

Such drawings have applications in various fields; for example, in channel
routing for integrated circuit design, see \cite{w-nrsil-91}.  Another
application is the visualization of metro maps and transportation networks
where some lines (railway tracks or roads) might partially overlap. A natural
goal is to draw the lines along their common subpaths so that an individual
line is easy to follow; minimizing the number of bends of a line and avoiding
unnecessary crossings between lines are natural criteria for map readability
(see Fig.~1(b) of~\cite{argyriou09}).  Much recent research in the graph
drawing community is devoted to edge bundling.  In this setting, drawing the
edges of a bundle with the minimum number of crossings and bends occurs as a
subproblem~\cite{pnbh-ifmpg-11}.

Let $S_n$ be the symmetric group of permutations $\pi=[\pi(1),\ldots,\pi(n)]$
on $\{1,\ldots,n\}$. The \df{identity permutation} is $[1, \dots, n]$, and
the \df{swap} $\sigma(i)$ transforms a permutation $\pi$ into $\pi \cdot
\sigma(i)$ by exchanging its $i$th and $(i+1)$th elements. (Equivalently,
$\sigma(i)$ is the transposition $(i,i+1)\in S_n$, and $\cdot$ denotes
composition.) Two permutations $a$ and $b$ of $S_n$ are \df{adjacent} if $b$
can be obtained from $a$ by swaps $\sigma(p_1), \sigma(p_2), \dots,
\sigma(p_k)$ that are not overlapping, that is, such that $|p_i-p_j|\ge 2$
for $i\ne j$.  A \df{tangle} is a finite sequence of permutations in which
each two consecutive permutations are adjacent.  An example of a tangle is
given in Fig.~\ref{example}. The associated drawing is composed of polylines
with vertices in $\mathbb{Z}^2$, whose segments can be vertical, or have
slopes of $\pm 45^{\circ}$ to the vertical.  The polyline traced by element
$i\in\{1,\ldots,n\}$ is called \df{path} $i$.  Note that by definition all
the path crossings occur at right angles.  We say that a tangle $T$
\df{solves} the permutation $\pi$ (or simply that $T$ is a tangle for $\pi$)
if the tangle starts from $\pi$ and ends at the identity permutation.

We are interested in tangles with informative and aesthetically pleasing
drawings. Our main criterion is to keep the paths straight by using only a
few turns.  A \df{corner} of path $i$ is a point at which it changes its
direction from one of the allowed directions (vertical, $+45^\circ$, or
$-45^\circ$) to another.  A change between $+45^\circ$ and $-45^\circ$ is
called a \df{double corner}.  We are interested in the total number of
corners of a tangle, where corners are always counted with multiplicity (so a
double corner contributes $2$ to the total).  By convention we require that
paths start and end with vertical segments.  In terms of the sequence of
permutations this means repeating the first and the last permutations at
least once each as in Fig.~\ref{example} (a).

Another natural objective is to minimize path crossings.  We call a tangle
for $\pi$ \df{simple} if it has the minimum number of crossings among all
tangles for $\pi$.  This is equivalent to the condition that no pair of paths
cross each other more than once, and this minimum number equals the {\em
inversion number} of $\pi$.  A simple tangle has no double corner, since that
would entail an immediate double crossing of a pair of paths.

In general, minimizing corners and minimizing crossings are conflicting
goals. For example, let $n=4k$ and $k\ge 4$ and consider the permutation
$$\pi=(2k,3,2,5,4,\dots,2k-1,2k-2,1,\quad 4k,2k+3,2k+2,\dots,4k-1,4k-2,2k+1).$$
It is not difficult to check that the minimum number of corners in a tangle
for $\pi$ is $4n-8$, while the minimum among simple tangles is $5n-20$, which
is strictly greater -- see Figure~\ref{fig:ex2} for the case $k=4$. Our focus
in this article is on two special classes of permutations for which corners
and crossings can be minimized simultaneously. The first is relatively
straightforward, while the second turns out to be much more subtle.
\begin{figure}[t]
\centering
\includegraphics[scale=.55]{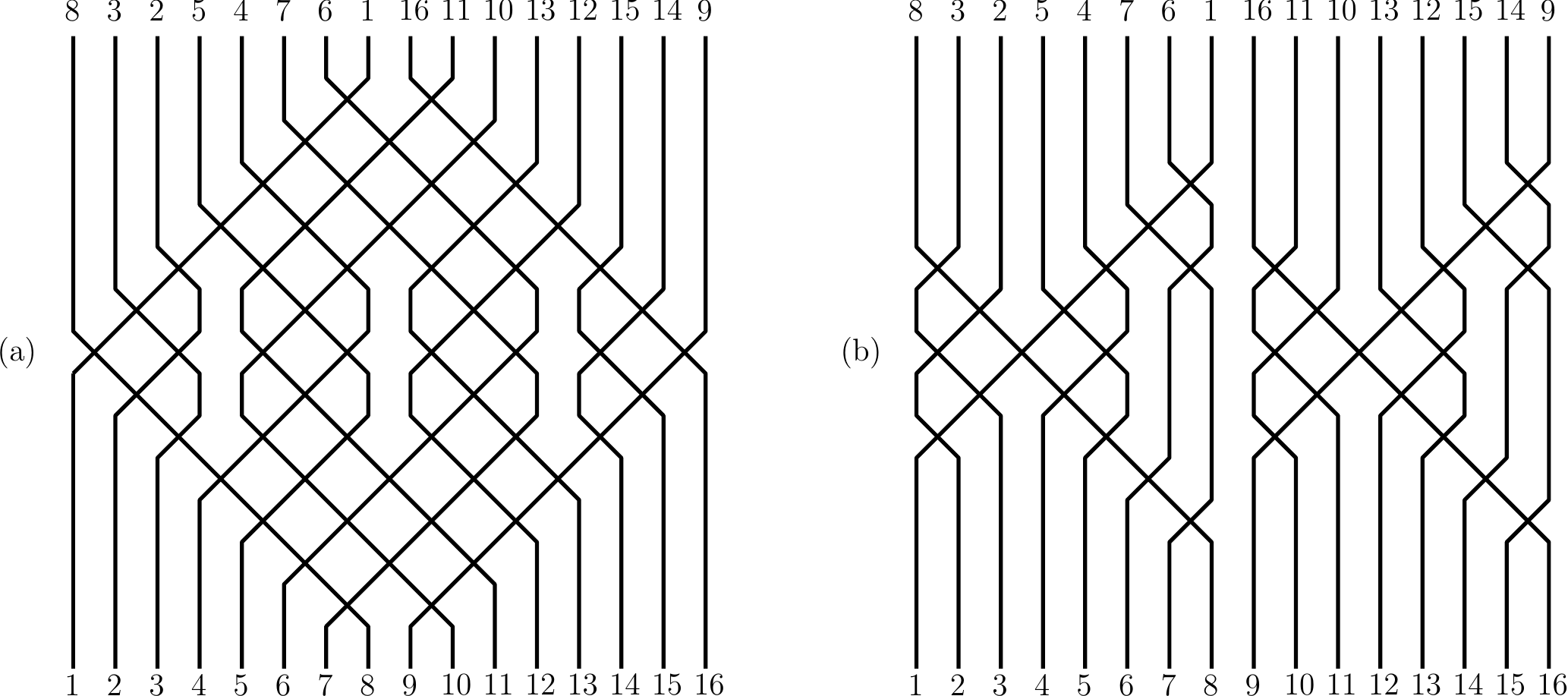}
\caption{(a) A tangle with $56$ corners.
(b) Every {\em simple} tangle for the same permutation has at least $60$ corners.}
\label{fig:ex2}
\end{figure}

One may ask the following interesting question. Is there an efficient
algorithm for finding a (simple) tangle with the minimum number of corners
solving a given permutation?  We do not know whether there is a
polynomial-time algorithm, either with or without the requirement of
simplicity.  We study this question in a forthcoming paper by the same authors and
we present approximation algorithms.  Here we give polynomial-time exact
algorithms for special classes of permutations.

Even the task of determining whether a given tangle has the minimum possible
corners among tangles for its permutation does not appear to be
straightforward in general (and likewise if we restrict to simple tangles).
However, in certain cases, such minimality is indeed evident, and we focus on
two such cases.  Firstly, we call a tangle \df{direct} if each of its paths
has at most $2$ corners (equivalently, at most one non-vertical segment).
Note that a direct tangle is simple.  Furthermore, it clearly has the minimum
number of corners among all tangles (simple or otherwise) for its
permutation.

We can completely characterize permutations admitting direct tangles.  We say
that a permutation $\pi\in S_n$ \df{contains} a pattern $\mu\in S_k$ if there
are integers $1 \le i_1 < i_2 < \dots < i_k \le n$ such that for all $1 \le r
< s \le k$ we have $\pi(i_r) < \pi(i_s)$ if and only if $\mu(r) < \mu(s)$;
otherwise, $\pi$ \df{avoids} the pattern.

\begin{theorem}
\label{thm:2corner} A permutation has a direct tangle if and only if it is
$321$-avoiding.
\end{theorem}

Our proof will yield a straightforward algorithm that constructs a direct tangle
for a given $321$-avoiding permutation.

Our second special class of tangles naturally extends the notion of a
direct tangle, but turns out to have a much richer theory.  A \df{segment} is
a straight line segment of a path between two of its corners; it is an
\df{L-segment} if it is oriented from north-east to south-west, and an
\df{R-segment} if it is oriented from north-west to south-east.  We call a
tangle \df{perfect} if it is simple and each of its paths has at most one
L-segment and at most one R-segment.
Any perfect tangle has the minimum possible corners among all tangles solving
its permutation, and indeed it has the minimum possible corners on path $i$,
for each $i=1,\dots, n$.  To see this, note that if $i$ has an L-segment in a
perfect tangle for $\pi$ then there must be an element $j>i$ with
$\pi(i)<\pi(j)$ (whose path crosses this L-segment). Hence, a L-segment must
be present in any tangle for $\pi$. The same argument applies to R-segments.
We call a permutation \df{perfect} if it has a perfect tangle.
%
%
\begin{theorem}
\label{thm:perfect} There exists a polynomial-time algorithm that determines
whether a given permutation is perfect, and, if so, outputs a perfect tangle.
\end{theorem}

A straightforward implementation of our algorithm takes $O(n^5)$ time, but we believe this
can be reduced to $O(n^3)$, and possibly further.
Our proof of Theorem~\ref{thm:perfect} involves an explicit characterization
of perfect permutations, but it is considerably more complicated than in the
case of direct tangles.  We will introduce the notion of a {\em marking},
which is an assignment of symbols to the elements $1,\ldots,n$ indicating the
directions in which their paths should be routed.  We will prove that a
permutation is perfect if and only if it admits a marking satisfying a {\em
balance} condition that equates numbers of elements in various categories.
Finally we will show that the existence of such a marking can be decided by
finding a maximum vertex-weighted matching \cite{Spencer84} in a certain
graph with vertex set $1,\ldots,n$ constructed from the permutation.

The number of perfect permutations in $S_n$ grows only exponentially with $n$
(see Section~\ref{sect:perfect}), and is therefore $o(|S_n|)$. Nonetheless,
perfect permutations are very common for small $n$: all permutations in $S_6$
are perfect, as are all but $16$ in $S_7$, and over half in $S_{13}$.


\paragraph{Related work.}
We are not aware of any other study on the number of corners in a tangle. To
the best of our knowledge, the problems formulated here are new. Wang
in~\cite{w-nrsil-91} considered the same model of drawings in the field of
VLSI design. However, ~\cite{w-nrsil-91} targets, in our terminology, the
tangle height and the total length of the tangle paths. The algorithm
suggested by Wang is a heuristic and produces paths with many unnecessary
corners. 

The perfect tangle problem is related to the problem of drawing graphs in
which every edge is represented by a polyline with few bends. In our setting,
all the crossings occur at right angles, as in so-called
RAC-drawings~\cite{walter09}.  We are not aware of any other study where
these two criteria are considered together.

Decomposition of permutations into nearest-neighbour transpositions was
considered in the context of permuting machines and pattern-restricted
classes of permutations~\cite{Albert07}. In our terminology,
Albert~et.~al.~\cite{Albert07} proved that it is possible to check in
polynomial time whether for a given permutation there exists a tangle of
length $k$ (i.e.\ consisting of $k$ permutations), for a given $k$. Tangle
diagrams appear in the drawings of sorting networks~\cite{knuth73,Angel07}.
We also mention an interesting connection with change ringing (English-style
church bell ringing), where similar visualizations are used~\cite{bells}.
In the language of change ringing, a tangle with minimum corners is
``a link method that produces a given row with minimum of changes of direction".

\section{Preliminaries}
\label{sect:str}
We always draw tangles oriented downwards with the sequence of permutations
read from top to bottom as in Fig.~\ref{fig:ex2}.
%
%
%
The following notation will be convenient. We write $\pi=[\dots a\dots b\dots
c \dots]$  to mean that $\pi^{-1}(a)<\pi^{-1}(b)<\pi^{-1}(c)$, and
$\pi=[\dots ab\dots]$ to mean that $\pi^{-1}(a)+1=\pi^{-1}(b)$, etc. A pair
of elements $(a,b)$  is an \df{inversion} in a permutation $\pi \in S_n$ if
$a>b$ and $\pi = [\dots a \dots b \dots]$. The \df{inversion number}
$\inv(\pi)\in[0,\binom{n}{2}]$ is the number of inversions of $\pi$. The
following useful lemma is straightforward to prove.

\begin{lemma}
\label{lm:invx}
In a simple tangle for permutation $\pi$, a pair $(i, j)$ is an inversion in $\pi$ if
and only if some R-segment of path $i$ intersects some L-segment of path $j$.
\end{lemma}


\section{Direct Tangles}
\label{sect:2corner}

In this section we prove Theorem~\ref{thm:2corner}.
We need the following properties of 321-avoiding permutations.

\begin{lemma}
\label{lm:mono} Suppose $\pi, \pi'$ are permutations with
$\inv(\pi')=\inv(\pi)-1$ and $\pi' = \pi \cdot \sigma(i)$ for some swap
$\sigma(i)$. If $\pi$ is 321-avoiding then so is $\pi'$.
\end{lemma}

\begin{proof}
Let us suppose that elements $i,j,k$ form a 321-pattern in $\pi'$. Then $(i,
j)$ and $(j, k)$ are inversions in $\pi'$. Inversions of $\pi'$ are
inversions of $\pi$, hence, elements $i, j, k$ form a 321-pattern in $\pi$.
\end{proof}

\begin{lemma}
\label{lm:no-LR} In a simple tangle solving a 321-avoiding permutation, no
path has both an L-segment and an R-segment.
\end{lemma}

\begin{proof}
Consider a simple tangle solving a 321-avoiding permutation $\pi$. Suppose
path $j$ crosses path $i$ during $j$'s R-segment and crosses path $k$ during
$j$'s L-segment.  By Lemma~\ref{lm:invx} we have $\pi=[\dots k \dots j \dots
i]$ while $i < j < k$, giving a 321-pattern, which is a
contradiction.\end{proof}

We say that a permutation $\pi\in S_n$ has a \df{split} at location $k$ if
$\pi(1), \ldots, \pi(k) \in \{1, \ldots, k\}$, or equivalently if
$\pi(k+1), \ldots, \pi(n) \in \{k+1, \ldots, n\}$.


\begin{proof}[Theorem~\ref{thm:2corner}]
To prove the ``only if'' part, suppose that tangle $T$ solves a permutation
$\pi$ containing a $321$-pattern.  Then there are $i < j < k$ with
$\pi=[\dots k \dots j \dots i]$.  Hence by Lemma~\ref{lm:invx}, $j$ has an
L-segment and an R-segment, so $T$ is not direct.

We prove the ``if'' part by induction on the inversion number of the
permutation. If $\inv(\pi) = 0$ then $\pi$ is the identity permutation, which
clearly has a direct tangle. This gives us the basis of induction.

Now suppose that $\pi$ is 321-avoiding and not the identity permutation, and
that every 321-avoiding permutation (of every size) with inversion number
less than $\inv(\pi)$ has a direct tangle. There exists $s$ such that $\pi(s)
> \pi(s+1)$; fix one such. Note that $(\pi(s), \pi(s+1))$ is an inversion of
$\pi$; hence, the permutation $\pi' := \pi \cdot \sigma(s)$ has $\inv(\pi') =
\inv(\pi) - 1$, and is also 321-avoiding by Lemma~\ref{lm:mono}. By the
induction hypothesis, let $T'$ be a direct tangle solving $\pi'$.

First perform a swap $x$ in position $s$ exchanging elements $\pi(s)$ and
$\pi(s+1)$, and draw it as a cross on the plane with coordinates $(s, h)$,
where $h\in \ZZ$ is the height ($y$-coordinate) of the cross (chosen
arbitrarily).  We assume that the position axis increases from left to right
and the height axis increases from bottom to top. Then draw the tangle $T'$
below the cross. This gives a tangle solving $\pi$, which is certainly
simple. We will show that the heights of swaps may be adjusted to make the
new tangle direct. To achieve this, the L-segment and R-segment comprising
the swap $x$ must either extend existing segments in $T'$, or must connect to
vertical paths having no corners in $T'$. We consider two cases.
\begin{figure}[t]
    \center
    \subfigure[]{
    \includegraphics[width=3.7cm]{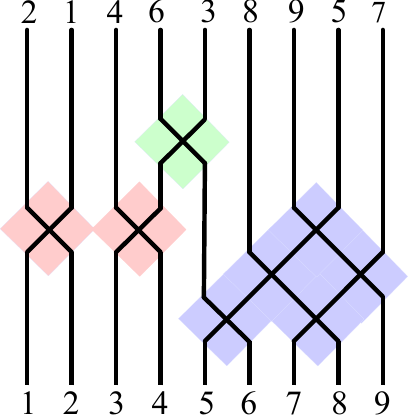}}
\qquad 
	\subfigure[]{
	\includegraphics[width=3.7cm]{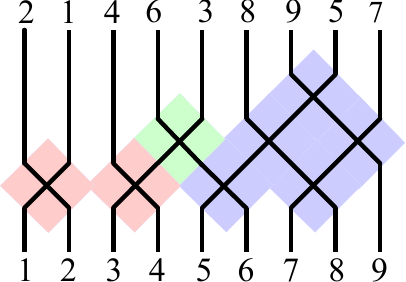}}
    \caption{Shifting two sub-tangles ($T_1$ is red, $T_2$ is blue)
    upward to touch the initial swap
     $x$ (green) in position $s=4$.}
    \label{fig:shift}
\end{figure}

\paragraph{Case 1.}
Suppose that $\pi'$ has a split at $s$. Then $T'$ consists of a tangle
    $T_1$ for the permutation $[\pi'(1), \dots, \pi'(s)]$ together with
    another tangle $T_2$ for $[\pi'(s+1), \dots, \pi'(n)]$; see
    Fig.~\ref{fig:shift}. Starting with $T_1$ drawn below $x$,
    simultaneously shift all the swaps of $T_1$ upward until one of them
    touches $x$; in other words, until $T_1$'s first swap in position
    $s-1$ occurs at height $h-1$. Or, if $T_1$ has no swap in position
    $s-1$, no shifting is necessary. Similarly shift $T_2$ upward (if
    necessary) until it touches $x$ from the right side. This results in
    a direct tangle.

\paragraph{Case 2.}
Suppose that $\pi'$ has no split at $s$. Let $T'$ be any direct tangle for  $\pi'$, and
again shift it upward until it touches $x$, resulting in a tangle $T$ for
$\pi$. Write $h_j$ for the height of the topmost swap in position $j$ in
$T'$, or let $h_j = -\infty$ if there is none. We claim that $h_{s-1} =
h_{s+1} = h_s+1 > -\infty$, which implies in particular that $1 < s < s+1
< n$. Thus $T'$ has swaps in the positions immediately left and right of
$x$, both of which touch $x$ simultaneously in the shifting procedure as
in Fig.~\ref{fig:cases}(a), giving that $T$ is direct  as required. To
prove the claim, first note that $h_s > -\infty$ since $\pi'$ has no
split at $s$. Therefore $\max\{h_{s-1},h_{s+1}\} > h_s$, otherwise $T$
would not be simple, as in Fig.~\ref{fig:cases}(b). Thus, without loss of
generality suppose that $h_{s-1} > h_s$ and $h_{s-1}\geq h_{s+1}$. Then
$h_s=h_{s-1}-1$, otherwise some path would have more than 2 corners in
$T'$, specifically, the path of the element that is in position $s$ after
$h_{s-1}$; see Fig.~\ref{fig:cases}(c) or (d). Now suppose for a
contradiction that $h_{s+1} < h_{s-1}$, which includes the possibility
that $h_{s+1}=-\infty$, perhaps because $s+1 = n$. Then in the new tangle
$T$, path $\pi(s)$ contains both an L-segment and an R-segment as in
Fig.~\ref{fig:cases}(e), which contradicts Lemma~\ref{lm:no-LR}.
\end{proof}
\begin{figure}[t]
    \center
	\includegraphics[height=3.5cm]{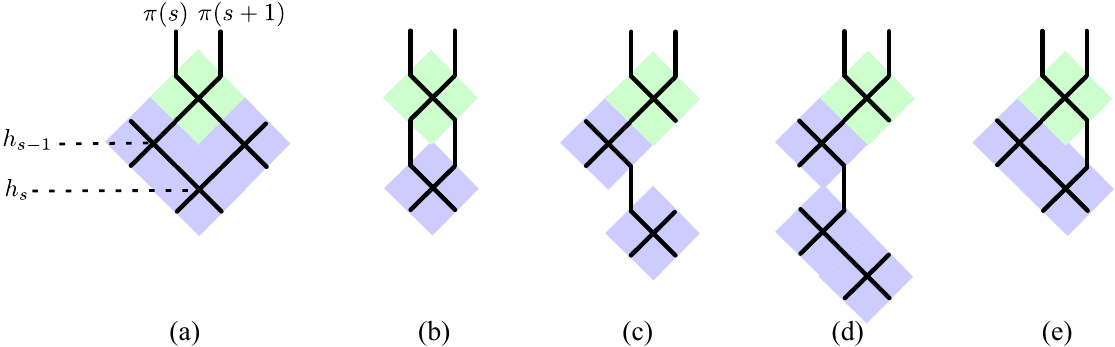}
    \caption{(a)~The tangle $T'$ (blue) touches the swap $x$ (green) on both sides.
    (b)--(e)~Various impossible configurations for the proof.}
    \label{fig:cases}
\end{figure}

The proof of the Theorem~\ref{thm:2corner} easily yields an algorithm that
returns a direct tangle  for $\pi\in S_n$ if one exists, and otherwise stops.
This algorithm can be implemented so as to run in $O(n^2)$ time.
With a suitable choice of output format, this can in fact be improved to $O(n)$.
\footnote{See the demo at
\url{http://www.utdallas.edu/~besp/very-nice-tangle.php}}

\section{Perfect Tangles}
\label{sect:perfect}

In this section we give our characterization of perfect permutations.

Given a permutation $\pi\in S_n$, we introduce the following classification
scheme of elements $i\in\{1,\ldots,n\}$. The scheme reflects the possible
forms of paths in a perfect tangle, although the definitions themselves are
purely in terms of the permutation. We call $i$ a \df{right} element if it
appears in some inversion of the form $(i,j)$, and a \df{left} element if it
appears in some inversion $(j,i)$. We call $i$ \df{left-straight} if it is
left but not right, \df{right-straight} if it is right but not left, and a
\df{switchback} if it is both left and right.

In order to build a perfect tangle we use a notion of marking. A \df{marking}
$M$ is a function from the set $\{1, \dots , n\}$ to strings of letters $L$
and $R$. For any tangle $T$, we associate a corresponding marking $M$ as
follows.  We trace the path $i$ from top to bottom; as we meet an L-segment
(resp.\ R-segment), we append an $L$ (resp.\ $R$) to $M(i)$. Vertical
segments are ignored this purpose, hence a vertical path with no corners is
marked by an empty sequence $\emptyset$. For example, $M(3)=R$ and $M(13)=LR$
in Fig.\ \ref{example} (c).  A marking corresponding to a perfect tangle takes
only values $\emptyset$, $L$, $R$, $LR$, and $RL$.  We sometimes write
$M(i)=R\dots$ to indicate that the string $M(i)$ starts with $R$.

Given a permutation $\pi$ and a marking $M$, there does not necessarily exist
a corresponding tangle. However, we will obtain a necessary and sufficient
condition on $\pi$ and $M$ for the existence of a corresponding perfect
tangle. Our strategy for proving Theorem~\ref{thm:perfect} will be to find a
marking satisfying this condition, and then to find a corresponding perfect
tangle. We say that a marking $M$ is a \df{marking for} a permutation $\pi
\in S_n$ if (i)~$M(i)=L$ (respectively $M(i)=R$) for all left-straight
(right-straight) elements $i$, (ii)~$M(i)\in\{LR,RL\}$ for all switchbacks,
and (iii)~$M(i)=\emptyset$ otherwise.




To state the necessary and sufficient condition mentioned above, we need some
definitions. A quadruple $(a, b, c, d)$ is a \df{rec} in permutation $\pi$ if
$\pi=[\dots a\dots b\dots c \dots d\dots]$ and $\min\{a,b\}>\max\{c,d\}$. In
a perfect tangle, the paths comprising a rec form a rectangle; see
Fig.~\ref{fig:recs}
 (``rec'' is an abbreviation for rectangle).
Let $M$ be a marking for $\pi \in S_n$, and let $\rho$ be a rec $(a, b, c,
d)$ in $\pi$. We call $e$ a \df{left switchback} of $\rho$ if (i)~$M(e)=RL$,
(ii)~$\pi=[\dots a \dots e\dots b\dots]$, and (iii)~$c < e < d$ or $d < e <
c$. Symmetrically, we call $e$ a \df{right switchback} of $\rho$ if
$M(e)=LR$, and $\pi=[\dots c \dots e\dots d\dots]$, and $a < e < b$ or $b < e
< a$. A rec $(a, b, c, d)$ is \df{regular} if $a < b$ and $c < d$, otherwise
it is \df{irregular}. A rec is called \df{balanced} under $M$ if the number
of its left switchbacks is equal to the number of its right switchbacks; a
rec is \df{empty} if it has no switchbacks.


Here is our key definition. A marking $M$ for a permutation $\pi$ is called \df{balanced} if every regular rec
of $\pi$ is balanced and every irregular rec is empty under $M$.
%
\begin{figure}[t]
    \center
    \includegraphics[width=5cm]{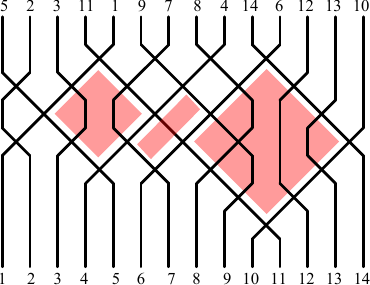}
    \caption{A permutation with a balanced marking.
    Some of the recs of the permutation are: $\rho_1=(5, 11, 1, 4)$,
     $\rho_2=(9, 7, 4, 6)$, $\rho_3=(11, 14, 6, 10)$;
    $\rho_1$ and $\rho_2$ are regular, while $\rho_2$ is irregular.
    Left switchbacks of rec $\rho_3$ are $8$ and $9$, right switchbacks are $12$ and $13$.
    The empty irregular rec $\rho_2$ has neither left nor right switchbacks.}
    \label{fig:recs}
\end{figure}
\begin{theorem}
\label{thm:main}
A permutation is perfect if and only if it admits a balanced marking.
\end{theorem}

The proof of Theorem~\ref{thm:main} is technical, and we postpone it to the
appendix.

Any permutation containing the pattern $[7324651]$ (for example) is not
perfect, since $4$ must be a switchback of one of the irregular recs $(7321)$
and $(7651)$. It follows by \cite{marcus-tardos,klazar} that the number of perfect
permutations in $S_n$ is at most $C^n$ for some $C>1$. Since direct tangles
are perfect it also follows from Theorem~\ref{thm:2corner} that the number is
at least $c^n$ for some $c>1$.

We note that Theorem~\ref{thm:main} already yields an algorithm for
determining whether a permutation is perfect in $\widetilde{O}(2^n)$
time\footnote{$\widetilde{O}$ hides a polynomial factor.}, by checking all
markings. In Section~\ref{sect:alg} we improve this to polynomial time.

%


\section{Recognizing Perfect Permutations}
\label{sect:alg}

We provide an algorithm for recognizing perfect permutations. The algorithm
finds a balanced marking for a permutation, or reports that such a marking
does not exist.  We start with a useful lemma.
\begin{lemma}
\label{lm:str}	Fix a permutation. For each right (resp., left) element $a$
there is a left-straight (right-straight) $b$ such that the pair $(a, b)$
(resp., ($b,a$)) is an inversion.
\end{lemma}
\begin{proof}
We prove the case when $a$ is right, the other case being symmetrical. Consider
the minimal $b$ such that $(a, b)$ is an inversion. By definition, $b$ is
left. Suppose that it is also a right element, that is, $(b, c)$ is an
inversion for some $c < b$. It is easy to see that $(a, c)$ is an inversion
too, which contradicts to the minimality of $b$.
\end{proof}

Recall that a marking is balanced only if (in particular) every regular rec
of the permutation is balanced under the marking.  We first show that this is
guaranteed even by balancing of recs of a restricted kind. We call a rec $(a,
b, c, d)$ of a permutation $\pi$ \df{straight} if $a, b, c$, and $d$ are
straight elements of $\pi$. A marking for a permutation $\pi$ is called
\df{s-balanced} if every straight rec is balanced and every
irregular rec is empty under the marking.

\begin{lemma}
\label{lm:sb}
Let $M$ be a marking of a permutation $\pi$. Then $M$ is
balanced if and only if it is s-balanced.
\end{lemma}

\begin{proof}
The ``if'' direction is immediate, so we turn to the converse.  Let $M$ be an
s-balanced marking and $\rho = (a, b, c, d)$ be a regular rec of $\pi$. We
need to prove that $\rho$ is balanced under $M$. If $\rho$ is straight then
$\rho$ is balanced by definition. Let us suppose that $\rho$ is not straight.
Then some $u\in \{a,b,c,d\}$ is not a straight element. Our goal is to show
that it is possible to find a new rec $\rho'$ in which $u$ is replaced with a
straight element so that the sets of left and right switchbacks of $\rho$ and
$\rho'$ coincide. By symmetry, we need only consider the cases $u=a$ and
$u=b$.

{\em Case $u=a$}. Let us suppose that $a$ is not straight. By
Lemma~\ref{lm:str}, there exists a right straight $e$ such that $(e, a)$ is
an inversion. Let us denote $\rho' = (e, b, c, d)$ and show that $\rho'$ has
the same switchbacks as $\rho$. Let $k$ be a left switchback of $\rho$; then
$M(k)=RL$, and $\pi=[\dots e \dots a \dots k \dots b \dots]$, and $c < k <
d$. By definition $k$ is a left switchback of $\rho'$. Let $k$ be a left
switchback of $\rho'$. If $\pi = [\dots e \dots k \dots a \dots b \dots]$
then the irregular rec $(e, a, c, d)$ has a left switchback, which is
impossible. Therefore, $\pi=[\dots e \dots a \dots k \dots b \dots]$ and $k$
is a left-switchback of $\rho$.

Let us suppose that $k$ is a right switchback of $\rho$, so $a < k < b$. If
$k < e$ then $k$ is a right switchback of the irregular $(e, a, c, d)$;
hence, $e < k < b$ and $k$ is a right switchback of $\rho'$. On the other
hand, if $k$ is a right switchback of $\rho'$ then $a < e < k < b$, which
means that $k$ is a right switchback of $\rho$.

{\em Case $u=b$}. Let us suppose that $b$ is not straight. By
Lemma~\ref{lm:str}, there exists a right straight $e$ such that $(e, b)$ is
an inversion. Let us denote $\rho' = (a, e, c, d)$ and show that $\rho'$ has
the same switchbacks as $\rho$.  Let $k$ be a left switchback of $\rho$. We
have $\pi=[\dots a \dots k \dots b \dots]$. Since $k$ is not a left
switchback of the irregular rec $(e, b, c, d)$, we have $\pi=[\dots a \dots k
\dots e \dots]$. Therefore, $k$ is a left switchback of $\rho'$.

Let $k$ be a right switchback of $\rho$. Then $a<k<b<e$, proving that $k$ is
a right switchback of $\rho'$. Let $k$ be a right switchback of $\rho'$. If
$b<k$ then $k$ is a right switchback of $(e, b, c, d)$, which is impossible.
Then $k<b$ and $k$ is a right switchback of $\rho$.
\end{proof}

We can restrict the set of recs guaranteeing the balancing of a permutation
even further. We call a pair $a, b$ of elements \df{right} (resp.\ \df{left})
\df{minimal} if $a$ and $b$ are right (left) straight elements of $\pi$, and
$a < b$, and there is no right (left) straight element $c$ such that
$\pi=[\dots a \dots c \dots b \dots]$. We call rec $\rho = (a,b,c,d)$
\df{minimal} in $\pi$ if $a, b$ is a right minimal pair and $c, d$ is a left
minimal pair; see Fig.~\ref{fig:algo1}. We call a marking for a permutation
\df{ms-balanced} if every minimal regular rec is balanced and every irregular
rec is empty under the marking.

\begin{lemma}
\label{lm:msb} Let $M$ be a marking of a permutation $\pi$. Then $M$ is
s-balanced if and only if it is ms-balanced.
\end{lemma}

Before giving the proof, we introduce some further notation. Let $\rho = (a,
b, c, d)$ be an arbitrary, possibly irregular, rec in $\pi$. Let us denote by
$\rho_\ell$ (resp.~$\rho_r)$ the set of switchbacks $i$ that can under some
marking be left (resp., right) switchbacks of $\rho$. Formally, $i\in
\rho_\ell$ if and only if $\pi = [\dots a \dots i \dots b \dots c \dots d
\dots]$ and either $c<i<d$ or $d<i<c$. (And $\rho_r$ is defined
symmetrically.) For a rec $\rho$ and marking $M$ let $\rho^M_\ell$
($\rho^M_r$) be the set of left (respectively, right) switchbacks of $\rho$
under $M$. Of course, $\rho^M_\ell \subseteq \rho_\ell$ and $\rho^M_r
\subseteq \rho_r$. It is easy to see from the definition that for two
different minimal recs $\rho$ and $\rho'$ we have $\rho_\ell \cap \rho'_\ell=
\emptyset$ and $\rho_r \cap \rho'_r=\emptyset$.

\begin{proof}[Lemma~\ref{lm:msb}]
It suffices to prove that if $M$ is ms-balanced then it is s-balanced.
Consider a straight rec $\rho = (a, b, c, d)$. Let $a=r_1, \dots, r_p=b$ be a
sequence of right straights in which each consecutive pair $r_i, r_{i+1}$ is
right minimal.  Define left straights $c=\ell_1, \dots, \ell_q=d$ similarly.
Let $D$ be the set of all recs of the form $(r_i, r_{i+1}, \ell_j,
\ell_{j+1})$ for $1\le i<p$ and $1\le j<q$. Notice that all recs of $D$ are
minimal. By definition of rec switchbacks, we have $ \rho^M_\ell =
\bigcup_{u\in D} u^M_\ell$ and $\rho^M_r = \bigcup_{u\in D} u^M_r$. Since
every rec $u \in D$ is balanced and for every pair $u,v \in D$ of different
recs $u^M_\ell \cap v^M_\ell= u^M_r \cap v^M_r=\emptyset$, we have
$|\rho^M_\ell| = |\rho^M_r|$; that is, $\rho$ is balanced under $M$.
\end{proof}

\newcommand{\RR}{\mathfrak{R}}
\newcommand{\LL}{\mathfrak{L}}
\newcommand{\II}{\mathfrak{I}}

Let us show how to construct an ms-balanced marking. For a permutation $\pi$,
let $\II_\ell=\bigcup \{ \rho_\ell : \rho \text{ is an irregular rec in } \pi
\}$ and $\RR_\ell= \bigcup \{\rho_\ell :\rho \text{ is a regular rec in
}\pi\}$, and define $\II_r,\RR_r$ similarly.
Our algorithm\footnote{The demo is at
\url{https://sites.google.com/site/gdcgames/nicetangle}} is based on finding
a maximum vertex-weighted matching, which can be done in polynomial
time~\cite{Spencer84}.

\begin{figure}[t]
    \center
    \subfigure[]{
    \includegraphics[width=5cm]{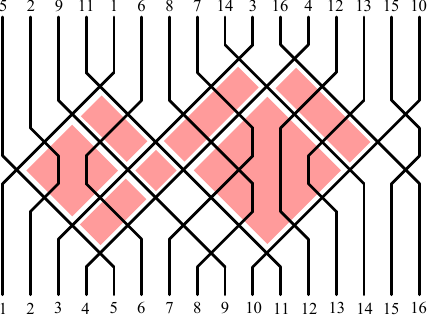}
    \label{fig:algo1}}
\qquad 
	\subfigure[]{
	\includegraphics[width=4cm]{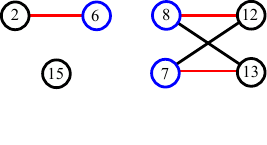}
    \label{fig:algo2}}
    \caption{(a)~A perfect tangle for a permutation with 7 minimal straight recs (shown red).
    (b)~The graph constructed in \emph{Step 3} of our algorithm.
    Here, $\II_\ell = \emptyset$, $\II_r = \{6\}$, $\RR_\ell = \{6, 7, 8, 12, 13, 15\}$, and $\RR_r = \{2, 7, 8\}$.
    The vertices of the set $F = \{6, 7, 8\}$ are shown blue. The red edges are the computed maximum matching.}
    \label{fig:algo}
\end{figure}

The algorithm (illustrated in Fig.~\ref{fig:algo}) inputs
a permutation $\pi$ and computes an ms-balanced marking $M$ for $\pi$ or
determines that such a marking does not exist. Initially, $M(i)$ is undefined
for every $i \in \{1,\dots, n\}$. The algorithm has the following steps.

\paragraph{\emph{Step 1}:} For every element $1\le i \le n$ that is neither
    left nor right, set $M(i) = \emptyset$. For every left straight $i$
    set $M(i) = L$. For every right straight $i$ set $M(i) = R$.

\paragraph{\emph{Step 2}:} If $\II_\ell \cap \II_r \neq \emptyset$  then
    report that $\pi$ is not perfect and stop. Otherwise, for every
    switchback $i \in \II_\ell$ set $M(i) = LR$; for every switchback $i
    \in \II_r$ set $M(i) = RL$.


\paragraph{\emph{Step 3.1}:}
Build a directed graph $G = (V, E)$ with $V = \RR_\ell
    \cup \RR_r$ and $E = \bigcup \{(\rho_\ell\setminus \II_\ell) \times
    (\rho_r\setminus \II_r) : \rho \text{ is a minimal rec in } \pi \}$.

\paragraph{\emph{Step 3.2}:}
Create a set $F \leftarrow (\RR_\ell \cap \RR_r)
    \cup (\II_\ell \cap \RR_r) \cup (\II_r \cap \RR_\ell)$. Create
    weights $w$ for vertices of $G$: if $i \in F$ then set $w(i)=1$,
    otherwise set $w(i)=0$.

\paragraph{\emph{Step 4}:}
Compute a maximum vertex-weighted matching $U$ on
    $G$ (viewed as an {\em unoriented} graph, ignoring the directions of edges) using weights $w$. If the total weight of $U$ is less than $|F|$
    then report that $\pi$ is not perfect and stop.

\paragraph{\emph{Step 5.1}:}
Assign marking based on the matching: for every
    edge $(i, j) \in U$ set $M(i) = RL$ provided $M(i) $ has not already been assigned,
	and $M(j) = LR$ provided $M(j) $ has not already been assigned.

\paragraph{\emph{Step 5.1}:}
For every switchback $1\le i \le n$ with still
    undefined marking, if $i \in \RR_\ell$ then set $M(i) = LR$, if $i \in \RR_r$ then set $M(i) = RL$,
	otherwise choose $M(i)$ to be $LR$ or $RL$ arbitrarily.  (Note that any $i \in \RR_\ell\cap \RR_r$
	was already assigned because of Steps 3.2 and 4.)

\vspace{2mm}
Let us prove the correctness of the algorithm.

\begin{lemma}
\label{lm:albal}
If the algorithm produces a marking then the marking is ms-balanced.
\end{lemma}

\begin{proof}
Let $M$ be a marking produced by the algorithm for a permutation $\pi$.
It is easy to see that $M(i)$ is defined for all $1\le i \le n$
(in \emph{Step 1} for straights and in \emph{Step 2} and \emph{Step 5} for switchbacks).
By construction, $M$ is a marking for $\pi$.

Let us show that $M$ is ms-balanced. Consider an irregular rec $\rho$ of $\pi$,
and suppose that $i \in \rho_\ell$. Since $\rho_\ell \subseteq \II_\ell$, in
\emph{Step 2} we assign $M(i)=LR$, that is, $i\not \in \rho^M_\ell$.
Therefore, $\rho$ does not have left switchbacks under $M$. Similarly, $\rho$
does not have right switchbacks under $M$. Therefore, $\rho$ is empty.

Consider a regular minimal straight rec $\rho$ in $\pi$. Suppose that $i \in
\rho^M_\ell$. Then $M(i)=RL$ and $i\in \rho_\ell \subseteq \RR_\ell$. If $i
\in \II_r$ then $i\in \RR_\ell \cap \II_r \subseteq F$; hence $i$ is incident
to an edge in $U$. Since no directed edge of the form $(k, i)$ is included in $G$ in \emph{Step
3.1}, there exists $(i, k)\in U$ for some $k$. On the other hand, if $i \not
\in \II_r$ then string $RL$ was not assigned to $M(i)$ in \emph{Step 5.2}, nor in \emph{Step 2}.
Thus, it was assigned in \emph{Step 5.1}, and again $(i, k)\in U$ for some
$k$. By definition of $E$ we have $k\in \rho_r$, because $k$ cannot appear in $\rho'_r$ for
any other minimal $\rho'\neq \rho$.
The algorithm sets $M(k)=LR$ at \emph{Step 5.2}; it could not have previously set $M(k)=LR$
at \emph{Step 2} because $k\notin \II_r$ by the definition of $E$.
Thus $k \in \rho^M_r$.

By symmetry, an identical argument to the above shows that if $k \in \rho^M_r$ then $i \in
\rho^M_\ell$ for some $i$ satisfying $(i, k)\in U$.
Since $U$ is a matching,
we thus have a bijection between elements of
$\rho^M_\ell$ and $\rho^M_r$. Therefore, $\rho$ is balanced under $M$.
\end{proof}

\begin{lemma}
\label{lm:20}
Let $\pi$ be a perfect permutation. The algorithm produces
a marking for $\pi$.
\end{lemma}

\begin{proof}
Since $\pi$ is perfect, there is a balanced marking $M$ for $\pi$.
Since $M$ is balanced, all irregular recs are empty under $M$; hence, the algorithm
does not stop in \emph{Step 2}.
To prove the claim, we will create a matching in the graph $G$ with total weight $|F|$.

Let $\rho$ be a minimal rec in $\pi$. Since $\rho$ is balanced under $M$, we
have $|\rho^{M}_\ell| = |\rho^{M}_r|$. Hence, let $W_{\rho}$ be an arbitrary
matching connecting vertices of $|\rho^{M}_\ell|$ with vertices of
$|\rho^{M}_r|$. Of course, $|W_{\rho}| = |\rho^{M}_\ell|$. Let $W = \bigcup
\{W_{\rho} : \rho \text{ is a minimal rec in } \pi\}$. We show that every
element of set $F$ is incident to an edge of $W$.

Suppose $i \in \RR_\ell \cap \RR_r$. Since $i$ is a switchback in $\pi$, we
have $M(i)=RL$ or $M(i)=LR$. In the first case $i \in \rho^{M}_\ell$ and in
the second case $i \in \rho^{M}_r$ for some minimal rec $\rho$. Then $i$ is
incident to an edge from $W_{\rho}$.

Suppose $i \in F \setminus \{\RR_\ell \cap \RR_r\}$. Without loss of
generality, let $i \in \II_\ell \cap \RR_r$. Since $M$ is balanced, every
irregular rec has no switchbacks and hence $M(i)=LR$. Thus, $i \in
\rho^{M}_r$ for some minimal rec $\rho$, and $i$ is incident to an edge of
$W_{\rho}$.

Therefore, every vertex of $F$ is incident to an edge of the matching $W$, which
means that the total weight of $W$ is $|F|$.
\end{proof}

Theorem~\ref{thm:perfect} follows directly from Lemmas~\ref{lm:albal}
and~\ref{lm:20} and Theorem~\ref{thm:main}.

\section{Conclusion}
\label{sect:conclude} In this paper we give algorithms for producing optimal
tangles in the special cases of direct and perfect tangles, and for
recognizing permutations for which this is possible.  The following questions
remain open.  (i) What is the complexity of determining the tangle with
minimum corners for a given permutation?  (ii) What is the complexity if the
tangle is required to be simple?  (iii) What is the asymptotic behavior of
the maximum over permutations $\pi\in S_n$ of the minimum number of corners
among simple tangles solving $\pi$?



\paragraph{Acknowledgements.} We thank Omer Angel, Franz Brandenburg, David Eppstein, Martin
Fink, Michael Kaufmann, Peter Winkler, and Alexander Wolff for fruitful
discussions about variants of the problem.

\bibliographystyle{abbrv}
\bibliography{nice-tangles}

\newpage
\appendix

\begin{figure}[htb]
\begin{center}
\includegraphics[scale=.4]{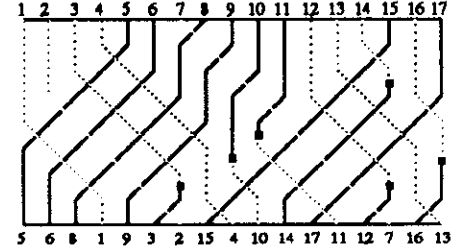}
\end{center}
\caption{An example of a channel routing from \cite[Figure 9]{w-nrsil-91}.}
\label{fig:wang}
\end{figure}

\section{Proof of ``if'' direction of Theorem~\ref{thm:main}}

We introduce the following additional definition. We say that a marking $M$
for $\pi$ is \df{aligned} with $\pi$ if, for every pair of elements $a, b$
with $\pi=[\dots ab\dots]$ and $M(a) = R\dots$ and $M(b) = L\dots$, we have
that $(a,b)$ is an inversion.

\begin{lemma}
\label{lm:X} Let $M$ be a balanced marking for $\pi \in S_n$. There exists an
aligned balanced marking $M'$ for $\pi$.
\end{lemma}

By Lemma \ref{lm:X}, the following theorem implies the ``if" direction  of
Theorem~\ref{thm:main}.

\begin{theorem}
\label{thm:xc} Let $M$ be a balanced and aligned marking for a permutation
$\pi$. Then $\pi$ has a perfect tangle with marking $M$.
\end{theorem}

The proof of Theorem~\ref{thm:xc} will be by induction on inversion number
similar to the proof of Theorem~\ref{thm:2corner}.  To handle the various
cases of the induction step, we will need the lemmas below, which will require
the following assumptions. Let $a$ and $b$ be two elements such that
\begin{equation} \label{eq1}
a=\pi(s), \quad b=\pi(s+1), \quad \text{ and }(a,b) \text{ is an inversion}.
\end{equation}
Let $M$ be a balanced and aligned marking for $\pi$ such that
\begin{equation} \label{eq2}
M(a)=L\dots, \quad\text{and}\quad M(b)=R\dots .
\end{equation}
Let $\pi'=\pi\cdot \sigma(s)$ and $M'$ be a marking for $\pi'$ defined as
follows. If $i\notin\{a,b\}$ then $M'(i)=M(i)$.  Let
\begin{equation} \label{eq3}
M'(a)=\begin{cases} RL & \text{if $a$ is switchback for $\pi'$,} \\
L & \text{if $a$ is straight for $\pi'$,}\\ \emptyset & \text{otherwise,}\end{cases}
\end{equation}
and define $M'(b)$ symmetrically.


\begin{lemma}
\label{lm:M'bal} Let $a=\pi(s)$ and  $b=\pi(s+1)$ and $\pi' = \pi \cdot
\sigma(s)$. Let $M$ be a balanced aligned marking for $\pi$, and let $M'$ be
a marking for $\pi'$ defined as in \eqref{eq3}. Then $M'$ is balanced.
\end{lemma}

\begin{lemma}
\label{lm:M'align} Under the assumptions of Lemma \ref{lm:M'bal}, $M'$ is
aligned.
\end{lemma}

\begin{lemma}
\label{lm:height} Under the assumptions of Lemma \ref{lm:M'bal}, suppose $T'$
is any perfect tangle for $\pi'$.  If $\pi'$ does not have a split at $a$,
and if $M(a)=M'(a)$ and $M(b)=M'(b)$, then the heights of the topmost swap of
$a$ and the topmost swap of $b$ in the drawing of $T'$ are equal.
\end{lemma}

\begin{proof}[Theorem \ref{thm:xc}]
The proof will be by induction on inversion number.
The statement of the theorem holds for the identity permutation, giving us
the basis of induction.

Consider a permutation $\pi$ with $\inv(\pi) = k > 0$ and a balanced aligned
marking $M$ for $\pi$. Since $\pi$ is not the identity permutation, let $i$
be minimal such that $\pi(i) \ne i$. It is easy to see that $M(\pi(i))=R$.
Symmetrically, for the maximal $j$ with $\pi(j) \ne j$ we have $M(\pi(j))=L$.
Therefore, there exists $i\le s < j$ such that $M(\pi(s)) = R\dots$ and
$M(\pi(s+1)) = L\dots$. Let us denote $a=\pi(s)$ and $b=\pi(s+1)$. Since $M$
is aligned, $(a, b)$ is an inversion. Thus, $a$, $b$ and $M$ satisfy
(\ref{eq1}) and (\ref{eq2}).

The permutation $\pi' = \pi \cdot \sigma(s)$ has $\inv(\pi') = \inv(\pi) -
1$. We define a marking $M'$ using Equation (\ref{eq3}). By construction,
$M'$ is a marking for $\pi'$. By Lemma~\ref{lm:M'bal}, $M'$ is balanced, and
by Lemma~\ref{lm:M'align}, $M'$ is aligned. Therefore, by the induction
hypothesis, there is a perfect tangle $T'$ for $\pi'$ with marking $M'$. To
obtain tangle the $T$ for $\pi$ we swap elements $a$ and $b$, and draw the
swap as a cross $x$ on the plane.  We then draw the tangle $T'$ below the
cross giving us a simple tangle for $\pi$. If $\pi'$ has a split at $s$ then
using the same arguments as in the case~(i) of Theorem~\ref{thm:2corner} we
adjust the heights of the cross and $T'$ so that $T$ is perfect. Let us
assume that $\pi'$ has no split at $s$.

If $M'(a)=M(a)$ and $M'(b)=M(b)$ then by Lemma \ref{lm:height} the height of
the topmost swap of $a$ and the height of the topmost swap of $b$ in $T'$ are
equal; see Fig.~\ref{fig:perfect_proof}(a). Therefore by shifting the cross
$x$ it is easy to construct a perfect tangle $T$; see
Fig.~\ref{fig:perfect_proof}(b).

\begin{figure}[t]
    \center
	\includegraphics[height=2.5cm]{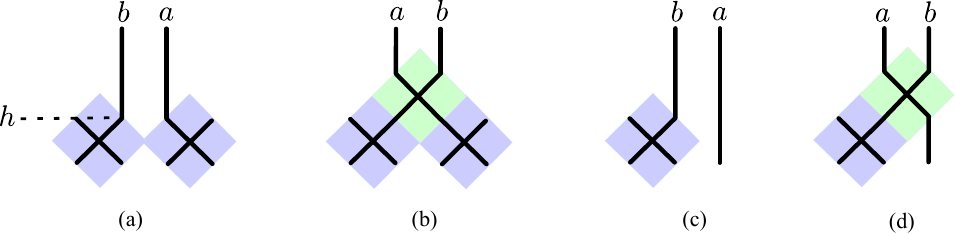}
    \caption{(a)~In the tangle $T'$ (blue) the heights of the topmost swaps
    of $a$ and $b$ are equal to $h$. (b)~The tangle $T$ with the cross $x$ (green).
    (c)--(d)~Adding the cross $x$ in the case $M'(a)=L$ and $M'(b)=L\dots$.}
    \label{fig:perfect_proof}
\end{figure}

Now let us consider the case when the marking of $a$ or $b$ changes. Without
loss of the generality, assume that $M'(a) \neq M(a)$. Since $\pi'$ does not
have a split at $s$, we have $M'(a)\ne \emptyset$ and therefore $M'(a)=L$.
Since $M'(b)\ne \emptyset$, we have $M'(b)=L\dots$, see
Fig.~\ref{fig:perfect_proof}(c). To obtain a perfect $T$ we draw the cross
$x$ immediately above the topmost swap of $b$; see
Fig.~\ref{fig:perfect_proof}(d).
\end{proof} 

We need another lemma to prove lemmas \ref{lm:X}, \ref{lm:M'bal}, and
\ref{lm:M'align}.


\begin{lemma}
\label{lm:b=a+1} Let $M$ be a balanced marking for a permutation $\pi =
[\dots a b \dots]$ with $M(a)=R\dots$, and $M(b)=L\dots$, and $a < b$. Then
$b = a+1$, and $M(a)=RL$, and $M(b)=LR$.
\end{lemma}

\begin{proof}
Let us show that $M(a)=RL$. Indeed, since $M(b)=L\dots$ there exists an
inversion $(c, b)$. Since $a < b$ we conclude that $(c, a)$ is an inversion
too, which means that $a$ is a left element. Therefore, $M(a)$ contains a
letter $L$, that is, $M(a)=RL$. Symmetrically, $M(b)=LR$.

Let us prove that $b = a+1$. Suppose for a contradiction that $b > a + 1 =:
c$. By Lemma~\ref{lm:str} there exists a right-straight $d$ such that $(d,
b)$ is an inversion, and there exists a left-straight $e$ such that $(a, e)$
is an inversion.  Hence, we have $e < a < c < b < d$. By symmetry we assume
that $\pi=[\dots d \dots a b \dots c\dots]$. Consider two cases according to
the order of $c$ and $e$ in $\pi$.

{\em Case 1}. Suppose that $\pi=[\dots d \dots a b \dots c \dots e\dots]$.
Notice that $(d, b, c, e)$ is a rec, and $a$ is a left switchback of the rec
under $M$. Since $d > b$ the rec is irregular, contradicting that $M$ is
balanced.

{\em Case 2}. Suppose that $\pi=[\dots d \dots a b \dots e \dots c\dots]$. We
have an irregular rec $(d, b, e, c)$ with a left-switchback $a$, which again
is a contradiction.
\end{proof}

\begin{proof}[Lemma \ref{lm:X}]
We show how to construct an aligned balanced marking $M'$ from $M$.
Initially, we set $M' = M$. Marking $M'$ is balanced at the initialization,
and we keep it balanced all the time.

We iteratively change $M'$ until it is aligned. If $M'$ is not aligned then
there exists a pair $a, b$ such that $\pi = [\dots a, b\dots]$, where $a <
b$, and $M'(a)=R\dots$, and $M'(b)=L\dots$. By Lemma~\ref{lm:b=a+1} we have
$b=a+1$, $M'(a)=RL$, and $M'(b)=LR$. We change $M'$ by setting $M'(a)=LR$ and
$M'(b)=RL$. Of course, $M'$ is still a marking of $\pi$.
It is easy to see that the iterative process is finite. Indeed, consider the
function $f(M')$ equal to the sum of elements with marking $LR$. On every
step $f(M')$ decreases by one, and $f(M')$ is clearly non-negative for any
marking $M'$. Therefore, after a finite number of iterations we obtain a
marking aligned with $\pi$.

\sloppypar
Let us prove that $M'$ remains balanced after an iteration. Consider any rec
$(u, v, w, x)$ of $\pi$. The rec is balanced on the previous iteration. To
become non-balanced it needs to lose or acquire a switchback. Then this
switchback is $a$ or $b$, since these are the only elements with new
markings. We consider three cases.

\begin{enumerate}[(i)]
\item Suppose neither $a$ nor $b$ is a member of the rec.  Without loss of generality we have $\pi =
    [\dots u \dots a b \dots v \dots  w \dots x \dots]$ (the other case with
	is
    symmetrical). Then the rec loses a left switchback $a$ and acquires a
    new left switchback $b$; hence, it remains balanced.

\item Suppose $u = a$. Then by definition of a rec we have $x < u = a <
    a+1 = b$. Hence, $b$ cannot be a left or a right switchback of the
    rec (either before or after the change), which means that its balance
    does not change.
\item The case $v=b$ is similar to (ii).
\end{enumerate}
All other cases follow by symmetry.
\end{proof}

\begin{proof}[Lemma \ref{lm:M'bal}]
Let $\rho=(u, v, w, x)$ be a rec of $\pi'$. Notice that all inversions of
$\pi'$ are inversions of $\pi$; hence, $\rho$ is a rec in $\pi$. We prove
that the set of its left (right) switchbacks under $M'$ in $\pi'$ is equal to
the set of its left (right) switchbacks under $M$ in $\pi$. Let us establish
the equality for the left switchbacks; the case with the right switchbacks is
symmetrical.

Consider a left switchback $c$ for $\rho$ in $\pi$ under $M$. Suppose for a
contradiction that $c$ is not a left switchback for $\rho$ in $\pi'$ under
$M'$.

\begin{enumerate}[(i)]
\item If $M(c) \ne M'(c)$ then by definition of $M'$ we have $c=a$ or
    $c=b$. Since $c$ is a left switchback for $\rho$ in $\pi$, we have
    $M(c)=RL$. Therefore, $c=a$ and $M(a)=M(c) \neq M'(c)=M'(a) = L$. By
    definition of a switchback, $a = c > \min(w, x)$. Since $b$ is a
    neighbor of $a$ in $\pi$, we have $b \not \in \{w, x\}$. Hence, $(a,
    \min(w, x))$ is an inversion in $\pi'$, so $a$ is a right element. We
    have a contradiction with $M'(a) = L$.

\item Suppose that $M(c) = M'(c)$. If $c$ is not a left switchback for
    $\rho$ in $\pi'$ and $M'(c) = RL$ then $\pi' = [\dots u \dots v c
    \dots]$. Hence, $c=a$ and $v=b$. In this case $a = c < \max(w, x) < v
    = b$, that is, $a < b$. We have a contradiction with $(a, b)$ being
    an inversion in $\pi$.
\end{enumerate}

Let us prove the converse: if $c$ is a left switchback for $\rho$ in $\pi'$
under $M'$ then $c$ is a left switchback for $\rho$ in $\pi$ under $M$. We
have $M'(c)=RL$, therefore, $M(c)=RL$ too. Suppose for a contradiction that
$c$ is not a left switchback for $\rho$ in $\pi$ under $M$. Then either $\pi
= [\dots u \dots v c \dots]$ or $\pi = [\dots c u \dots v \dots]$. In the
first case $a=v$ and $b=c$, which is a contradiction because $M(b)=L \dots
\ne M(c)=RL$. In the second case we have $a=c$ and $b=u$; hence, $a < u = b$,
which contradicts to the fact that $(a, b)$ is an inversion in $\pi$.
\end{proof}

\begin{proof}[Lemma \ref{lm:M'align}]
Suppose for a contradiction that $M'$ is not aligned. Without loss of
generality, assume that for some $c$ we have $\pi'=[\dots b a c \dots]$, with
$M'(a)=R\dots$, and $M'(c)=M(c)=L\dots$, and $(a, c)$ not an inversion. By
Lemma~\ref{lm:M'bal}, $M'$ is balanced; therefore, by Lemma~\ref{lm:b=a+1},
$c=a+1$. Since $a$ is right in $\pi'$, there exists an inversion $(a, d)$ in
$\pi'$ and in $\pi$. Since $c$ is left in $\pi$, there is an inversion $(e,
c)$ in $\pi$. Notice that $e \not \in \{a,b\}$ because $b < a < c$. Thus,
$\pi=[\dots e \dots a b c \dots d \dots]$. Then $(e, a, b, d)$ is a rec in
$\pi$. The rec is irregular since $a < e$. The element $c$ is a right
switchback of the rec, which contradicts $M$ being balanced.
\end{proof}

\begin{proof}[Lemma \ref{lm:height}]
Since $\pi'$ does not have a split at $s$, $T'$ has a swap at position $s$.
Let $c,d$ be the elements forming the topmost swap $x$ at position $s$ in
$T'$. Since $x$ is the topmost swap, $d$ starts to the right of $a$ in the
drawing of $T'$. Similarly, $c$ starts to the left of $b$; see
Fig.~\ref{fig:sh}(a). Hence, $\pi = [\dots c \dots a b \dots d\dots]$. Notice
that $M'(a)=M(a)=R\dots$; thus, path $a$ crosses path $d$ in $T'$ above $x$
and $a > d$. Similarly, path $b$ crosses path $c$ above $x$ and $b < c$.
Therefore, $(c, a, b, d)$ is a rec in $\pi$.
\begin{figure}[t]
\begin{center}
\includegraphics[height=4cm]{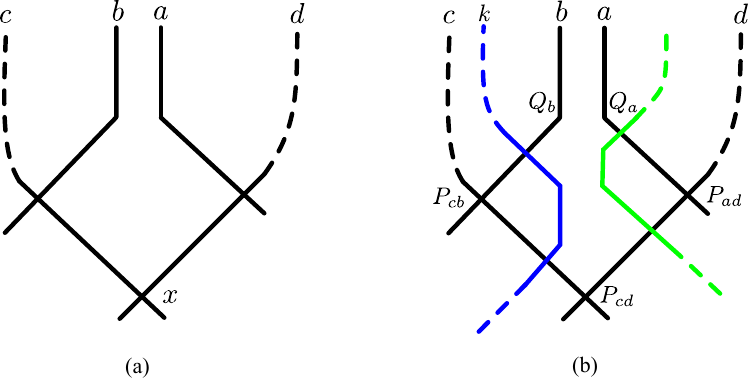}
\end{center}
\caption{(a)~A fragment of tangle $T'$ with the swap $x$.
(b)~The blue path $k$ is of type $1$; the green path is of type $2$. Points $Q_a$ and $Q_b$
have the same height by Lemma~\ref{lm:height}.}
\label{fig:sh}
\end{figure}

Let $P_{cb}$ be the intersection point of paths $c$ and $b$,  and
similarly define $P_{cd}$ and $P_{ad}$. Let $Q_b$ and $Q_a$ be the topmost
corners of paths $b$ and $a$ in $T'$; see Fig.~\ref{fig:sh}(b). To prove the
lemma, it sufficient to show that
\begin{equation}
|P_{cd}P_{cb}| + |P_{cb}Q_b| = |P_{cd}P_{ad}| + |P_{ad}Q_a|,
\end{equation}
where $|\cdot|$ denotes the length of a segment.

There are two types of paths crossing those line segments. Indeed, a path
crossing, for example, $P_{cd}Q_b$ at an inner point has to turn left before
reaching position $s+1$, otherwise it would create a swap higher than the
swap $x$. We say that path $k$ has type $1$ (type $2$) if it crosses segments
$P_{cd}P_{cb}$ and $P_{cb}Q_b$ ($P_{cd}P_{ad}$ and $P_{ad}Q_a$). Let us show
that a path $k$ has type $1$ (type $2$) if and only if $k$ is a left
switchback (right switchback) of the rec $(c, a, b, d)$ in $\pi$. We prove
the claim for type 1, the other case being symmetrical.

Consider a path $k$ of type $1$ as in Fig~\ref{fig:sh}(b). Since $T'$ is
simple, we have $\pi = [\dots c \dots k \dots a b \dots]$ and $(k, b)$ is an
inversion. Since $k$ crosses paths $c$ and $b$, we have $M(k)=RL$. That means
$k$ does not cross path $d$, that is, $k < d$, otherwise the marking of $k$
would be contain two $R$s. Therefore, $b < k < d$ and $k$ is a left
switchback of the rec.

From the other side, let us suppose that $k$ is a left switchback of the
regular rec $(c, a, b, d)$. Then $\pi=[\dots c \dots k \dots a b \dots]$ and
$M(k)=M'(k)=RL$. By the definition of a left switchback, we have $b < k < d$.
Hence, $(c, k)$ and $(k, b)$ are inversions of $\pi$ and path $k$ intersects
paths $c$ and $b$. Since $M'(k) = RL$, the path $k$ first crosses $b$ and
then $c$, that is, $k$ is of type $1$.

Since the rec is balanced under $M$, the number of paths of type $1$ is equal
to the number of paths of type $2$. Each path of type $1$ adds the same
amount to the left part of the equation $(1)$ as a path of type $2$ adds to
the right, so the equation holds.
\end{proof}

\section{Proof of ``only if'' direction of Theorem~\ref{thm:main}}

Let us prove the following useful lemma.

\begin{lemma}
\label{lm:type==sb} Let $T$ be a perfect tangle for $\pi$ with marking $M$.
Let $\rho = (a, b, c, d)$ be a rec in $\pi$ forming rectangle $P_\ell
P_tP_rP_b$ in the drawing of $T$ ($P_t$ is the topmost point of the
rectangle, $P_\ell$ is the leftmost, $P_r$ is the rightmost, and $P_b$ is the
lowest). Then path $k$ crosses segments $P_\ell P_t$ and $P_\ell P_b$ (resp.\
$P_rP_t$ and $P_rP_b$) if and only if $k$ is a left (resp.\ right) switchback
of $\rho$.
\end{lemma}

\begin{proof}
Let us prove the claim for the left switchbacks.  The other case is
symmetrical.

To prove the ``if'' part, let $k$ be a left switchback of $\rho$, see
Fig.~\ref{fig:balance2}(a). By definition of a left switchback, $k$ crosses
$u = \min(c, d)$. Notice that $P_\ell P_t$ belongs to the L-segment of $u$,
since otherwise $k$ would also cross $\max(c, d)$. Denote the intersection
point of $k$ and $u$ by $p$. We need to show that $p$ lies on the segment
$P_\ell P_t$.
\begin{figure}[t]
\begin{center}
\includegraphics[height=4cm]{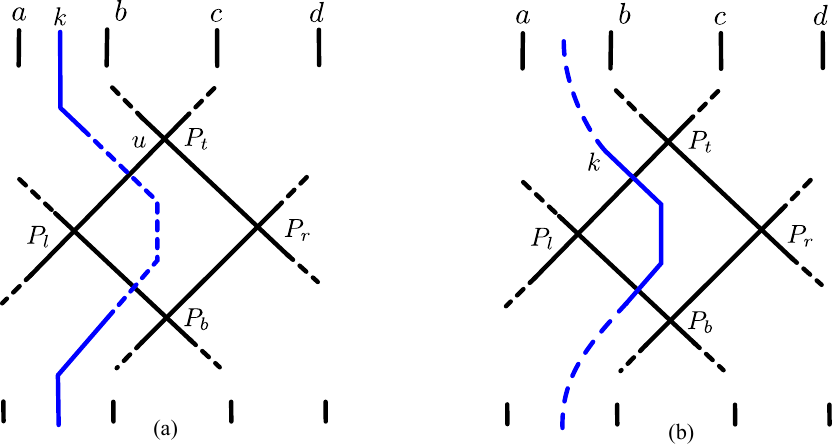}
\end{center}
\caption{(a)~Left switchback of the rec $(a, b, c, d)$ has to cross segments
 $P_\ell P_t$ and $P_\ell P_b$.
(b)~Path $k$ crossing $P_\ell P_t$ and $P_\ell P_b$ is necessarily a left
 switchback of $(a, b, c, d)$.}
\label{fig:balance2}
\end{figure}

Since $M(k)=LR$, path $k$ first crosses the L-segment of $u$ and then the
R-segment of $a$ (when followed from top to bottom of the tangle). Therefore,
$k$ crosses $P_\ell P_t$ above point $P_\ell$, that is, the intersection
point is higher than $P_\ell$. On the other hand, $k$ and $b$ do not cross,
which means that $p$ is below $P_t$. Hence, $p$ lies on $P_\ell P_t$.

By using symmetrical arguments it is easy to show that the intersection point
between paths $\min(a, b)$ and $k$ is on the segment $P_\ell P_b$. Therefore,
the left switchback $k$ crosses $P_\ell P_t$ and $P_\ell P_b$.

%

To prove the ``only if'' part of the claim, let $k$ be a path crossing
segments $P_\ell P_t$ and $P_\ell P_b$, see Fig.~\ref{fig:balance2}(b). If
$P_\ell P_t$ belongs to the L-segment of path $c$ then (i) $k > c$ since $k$
crosses $P_\ell P_t$, and (ii) $k < d$ since $k$ does not cross $d$
(otherwise, the marking of $k$ would contain two $L$s); that is, $c < k < d$.
If $P_\ell P_t$ belongs to the L-segment of path $d$ then by a symmetrical
argument $d < k < c$. It is also easy to see that the marking of $k$ is $RL$;
hence, $k$ is a left switchback of $\rho$.
\end{proof}


\begin{proof}[The ``only if'' direction of Theorem~\ref{thm:main}]
Let $T$ be a perfect tangle for a permutation $\pi$ and $M$ be the marking
corresponding to $T$. We need to prove that $M$ is balanced.

Let $\rho = (a, b, c, d)$ be a rec of $\pi$. It is easy to see that the
L-segments of paths $c$ and $d$ and the R-segments of paths $a$ and $b$ form
a rectangle in the drawing of $T$. We consider two cases.

\vspace{1mm}
\noindent {\bf Case 1.}
Suppose $\rho$ is regular. Let us show that the number of its right
    and left switchbacks under $M$ is the same. The L-segments of paths
    $c$ and $d$ cross the R-segments of paths $a$ and $b$. Let $P_{ac},
    P_{bc}, P_{ad}, $ and $P_{bd}$ be the crossing points of the
    segments. Since $\rho$ is regular, $P_{bc}$ is the topmost and
    $P_{ad}$ is the bottommost in the drawing of $T$; see
    Fig.~\ref{fig:balance}(a).

\begin{figure}[htb]
\begin{center}
\includegraphics[height=4cm]{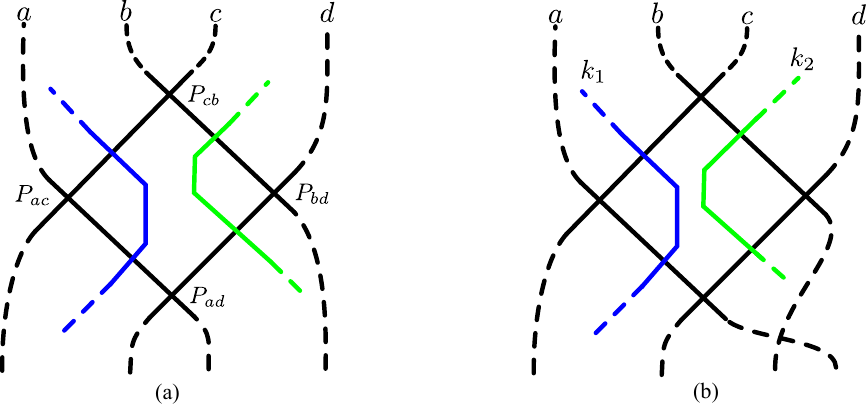}
\end{center}
\caption{(a)~In a perfect tangle every left switchback is of type $1$ (blue), and every right
switchback is of type $2$ (green).
(b)~Impossible configuration in a perfect tangle: path $k_2$ of type $2$ does not allow
paths $a$ and $b$ to cross each other ``below'' the rectangle.}
\label{fig:balance}
\end{figure}

We say that a path has type $1$ (resp.\ type $2$) in $T$ if it crosses
segments $P_{ac}P_{bc}$ and $P_{ac}P_{ad}$ (resp., $P_{bc}P_{bd}$ and
$P_{bd}P_{ad}$). By Lemma~\ref{lm:type==sb}, every left switchback of
$\rho$ is of type $1$, and every right switchback of $\rho$ is of type
$2$.  Since $P_{ac}P_{bc}P_{bd}P_{ad}$ is a rectangle, the number of paths
that intersect the side $P_{ac}P_{bc}$ equals the number that intersect the opposite side
$P_{bd}P_{ad}$.  A path may intersect neither, both, or only one of these two sides.  Those
that intersect only $P_{ac}P_{bc}$ are precisely the type 1 paths, while those that intersect only
$P_{bd}P_{ad}$ are precisely the type 2 paths.  Therefore
 in the drawing of
$T$, the number of paths of type $1$ is equal to the number of paths of
type $2$.  So, $\rho$ is balanced.

\vspace{1mm}
\noindent {\bf Case 2.}
Suppose $\rho$ is irregular. Let us show that it does not have any switchbacks under $M$.
Suppose for a contradiction that $\rho$ has a left switchback $k_1$.
The L-segments of $c$ and $d$ and the R-segments of $a$ and $b$ form a rectangle $P_\ell P_tP_rP_b$ in the drawing of $T$.
Then by Lemma~\ref{lm:type==sb}, $k_1$ crosses segments $P_\ell P_t$ and $P_\ell P_b$.
Hence, by the same argument as in Case 1 above,
there is a path $k_2$ of type $2$; the path is a right  switchback of $\rho$ by Lemma~\ref{lm:type==sb}.
Since the rec $\rho$ is irregular, at least one of the pairs of paths $(a,b)$ or $(c,d)$ has a crossing, which may occur above or below the
rectangle.  Without loss of generality, suppose that $a$ and $b$ cross below the rectangle, as in Fig.~\ref{fig:balance}(b).
Then, since $b<k_2<a$ by the definition of a rec, $k_2$ must cross $b$ twice, which is a contradiction.
\end{proof}

\end{document}